\documentclass[a4paper]{article}

\usepackage[shortlabels]{enumitem}
\usepackage[]{geometry}
\usepackage[utf8]{inputenc}
\usepackage{amssymb}
\usepackage{amsmath}
\usepackage{hyperref}
\usepackage{graphicx}
\usepackage{enumitem}
\usepackage{amsfonts}
\usepackage{amssymb}
\usepackage{makeidx}
\usepackage{tabularx}
\usepackage{float}
\usepackage[usenames, dvipsnames]{color}
\usepackage{url}
\usepackage[utf8]{inputenc}
\usepackage[english]{babel}
\usepackage[normalem]{ulem}
\usepackage{amsfonts}
\usepackage{dcolumn}
\usepackage{makeidx}
\usepackage{bm}
\usepackage[usenames]{color}
\usepackage{multirow}
\usepackage{graphicx}
\usepackage{hyperref}
\usepackage{tikz}
\usepackage{booktabs}
\usepackage{xcolor}
\usepackage{fancyhdr}
\usepackage{amsmath, amsthm, amssymb} 
\usepackage{multicol}
\usepackage{multicol}
\usepackage{caption}
\usepackage{amssymb,amsfonts,amssymb} 
\usepackage{amsmath}
\usepackage{tikz}
\usepackage{transparent}
\usepackage{soul}
\usepackage{cite}
\date{\today}

\newcommand{\GSTs}[0]{\mathcal S_{\alpha,\beta}^{a,b}}

\newcommand{\GSTIs}[0]{\widetilde{\mathcal S}_{\alpha,\beta}^{a,b}}

\numberwithin{equation}{section}
\usepackage{authblk}
\title{Sundman-like transformations and the NRT nonlinear Schrödinger equation}
\author[1]{P. R. Gordoa}
\author[1]{A. Pickering}
\author[1,2]{D. Puertas-Centeno}
\author[1]{E. V. Toranzo}

\affil[1]{Departamento de Matemática Aplicada, Ciencia e Ingeniería de los Materiales y Tecnología Electrónica, Universidad Rey Juan Carlos,
	28933 Móstoles (Madrid), Spain}
\medskip
\affil[2]{Data, Complex Networks and Cybersecurity Research Institute, Universidad Rey Juan Carlos, 28028 (Madrid), Spain}

\usepackage{blindtext}
\newtheorem{defi}{Definition}
\newtheorem{lemma}{Lemma}
\newtheorem{coro}{Corollary}
\newtheorem{remark}{Remark}
\newtheorem{teorema}{Theorem}
\newtheorem{prop}{Proposition}

\usepackage[T1]{fontenc}
\usepackage{palatino}
\begin{document}
	\maketitle

\begin{abstract}
We present a new generalization of the well-known power-type Sundman transformation, involving not only powers of the function but also of its derivative, along with its inverse. Our aim is to explore the use of such transformations in the derivation of solutions of ordinary differential equations and in the study of their properties. We then show their  usefulness in the framework of the nonlinear Nobre--Reigo-Monteiro--Tsallis (NRT) nonlinear Schr\"odinger equation. More precisely, we employ them to analyze a family of  ordinary differential equations which includes the Lorentzian solutions of the NRT-nonlinear Schrödinger equation for a constant potential.  Moreover, an explicit expression for the Lorentzian solitary wave solutions is given, for any real value of the non-linearity parameter $q$, in terms of a transformation depending on $q$ applied to the classical Lorentzian solution with $q=1$, i.e., we succeed in encapsulating the whole nonlinear behavior in the new transformations. In addition, the composition of this transformation with its inverse (with different parameters) allows us to perform a shift in the nonlinearity parameter $q$. Moreover, a certain subfamily of our generalized transformations, which perform a shift on the non-linearity parameter $q$ of the Lorentzian solutions, is found to have a group structure.  The same subfamily of transformations allows us, again, to perform a shift in the non-linearity parameter $q$, but in this case in the traveling wave solution for a free particle.
\end{abstract}

\section{Introduction}
In the study of nonlinear ordinary differential equations, a huge variety of techniques are available in order to obtain solutions, perhaps by reducing to a series of quadratures, or at least to obtain reductions in order, e.g., through the use of Lie symmetries, or via mappings to more manageable equations. One class of mapping often used is the power-type Sundman transformation \cite{Sud13}, see also \cite{DMS94}. Examples of the application of such transformations to physically interesting problems can be found, for example, in \cite{KS14,KS15,MR16,PG21,GPPT25}.

In this paper we define a generalization of the power-type Sundman transformation, involving not only powers of the function but also of its derivative. We also define its inverse transformation. We see here how these new transformations, which depend on four parameters, can be used in the study of ordinary differential equations. It turns out that for certain classes of autonomous equation of order $n\geq2$, the resulting transformed equations are also autonomous but of order $n-1$. This is then in contrast with the standard result when using Lie symmetries whereby the order of an autonomous equation can always be reduced by one, but at the price of passing to a non-autonomous equation. In this paper we use our generalized Sundman transformations to reduce a certain family of second order autonomous nonlinear ordinary differential equations to first order autonomous ordinary differential equations. It is worth mentioning that this new class of transformations also generalizes a pair of  biparametric transformations recently introduced in a totally different context, the framework of information theory~\cite{IP2025(c)}, and have been shown to be a useful tool in the study of informational inequalities~\cite{IP2025,IP2025(b)}. We now describe in more detail the problem tackled here.

In the framework of statistical mechanics, the extension of extensive statistical mechanics to a non-extensive formalism developed by Tsallis has given rise to an extraordinary amount of theoretical and experimental results since the 1980s. Later, this formalism inspired the introduction (amongst other equations) of a nonlinear version of the Schrödinger equation~\cite{NRT11}, namely the Nobre--Reigo-Monteiro--Tsallis (NRT) nonlinear Schr\"odinger equation, with an additional term $\Xi(\psi(x,t))$ proposed by Rego-Monteiro~\cite{REG20b}:
\begin{align}
\label{eq:Psi_eq+gen}
i\hbar \frac{\partial}{\partial t}\left[\psi(\vec{x},t) \right]	&= -\frac{1}{2-q}\frac{\hbar^2}{2m} \nabla^{2}\left[ \left(\psi(\vec{x},t)\right) ^{2-q} \right] +\Xi(\psi(x,t))\, , \quad q\in\mathbb{R}\setminus\{1,2\}\, .
\end{align} 
%
This equation has been studied for different particular functions $\Xi$ in~\cite{NRT12,TPDP13,PLA14,REG20a}. We  recall here that the probability density associated to the solution of the \textit{classical} Schrödinger equation is given  using the square of the modulus of the wave function, as established by Born himself  in the early stages of quantum physics. In stark contrast to this result, in the NRT formalism, the probability density is obtained  using the product of the wave function and an auxiliary, or slave, field~\cite{REG13}.
 
In the monodimensional case, for $\Xi(\psi(x,t))=\lambda\psi(x,t)^q$, Rego-Monteiro  found a Lorentzian solitary wave behavior for the energy density of a traveling-wave solution~\cite{REG20b}. 
This kind of solution is one focus of the work presented here. In particular, letting $z=i(kx-\omega t)$  in Eq.~\eqref{eq:Psi_eq+gen} with $\Xi(\psi(x,t))=\lambda\psi(x,t)^q$ yields~\cite{REG20b}
 \begin{align}
 \label{eq:Psi_eq+travRM}
 \hbar \omega \psi'(z) &= \frac{\hbar^2k^2}{2m}\left[\psi^{1-q}\psi''(z)+(1-q)\psi^{-q}(z)(\psi'(z))^2\right]+\lambda\psi(z)^q\, .
 \end{align} 
 Note that in the case $\lambda=0$ we recover the traveling-wave reduction studied by Bountis and Nobre~\cite{BN16} and whose general solution has  recently been analyzed  by the current authors~\cite{GPPT25} for any real value of the nonlinearity parameter $q$. Moreover, a closed-form expression for  the slave field and the probability density were also obtained through the use of a family of special functions called \textit{generalized trigonometric functions}~\cite{DRA99}. 
 
 In addition, we will also study in this paper the solutions of equations of the general form
   \begin{align}
  \label{eq:Psi_eq+travGen}
  \hbar \omega \psi'(z) &= \frac{\hbar^2k^2}{2m}\left[\psi^{1-q}\psi''(z)+(1-q)\psi^{-q}(z)(\psi'(z))^2\right]+\Xi(\psi, \psi')\, 
  \end{align}
  where 
  \begin{align}\label{eq:Gamma_cond}
\Xi(\psi, \psi')&=\hbar \omega \psi'(z)-\frac{(\psi')^2}{\psi^{q}}\,\overline g\left(\psi^\gamma \psi'\right)
  \end{align}
for some $\gamma\in\mathbb R$ and  where $\overline g$ is an arbitrary function. In the particular case where  $\gamma=-q$ and $\overline g(\xi)=\frac{\hbar \omega}\xi-\frac\lambda{\xi^2}$ in Eq.~\eqref{eq:Gamma_cond}, so that  $\Xi(\psi,\psi')=\lambda\psi(z)^q$,  we recover Eq.~\eqref{eq:Psi_eq+travRM}. This is the unique case in which the condition~\eqref{eq:Gamma_cond} is satisfied by $\Xi(\psi,\psi')=\Xi(\psi)$, i.e., where $\Xi$ is a function of $\psi$ only  Here we will  give the general solution of Eq.~\eqref{eq:Psi_eq+travGen} under the condition~\eqref{eq:Gamma_cond}, along with  some interesting mathematical properties relating  solutions for different values of the non-linearity parameter $q$.
 We also find, unexpectedly, that  the general solution of  Eq.~\eqref{eq:Psi_eq+travRM} can be written as a Sundman transformation of the classical (or linear) solution, i.e, corresponding to $q=1$. This fact allows us to capture  the whole non-linearity character in the transformation itself. Indeed, the transformation plays the role of a general shift relation between solutions for  different values of the parameter~$q.$ 


The layout of the paper is as follows. In Section 2 we define our generalized Sundman transformation (GST) and its inverse (IGST), which depend on four parameters, and discuss some of their properties. In particular, we find that in a special case where one of the parameters is zero, the GSTs form a  group under composition. For this special case we also consider applications to generalized trigonometric functions. Also in Section 2, we discuss how GSTs can be employed in order to derive solutions of certain classes of autonomous ordinary differential equations. In Section 3, we consider their application to precisely such a class of second order equations, thus obtaining results to be used later, as well as presenting a simple explicit example. It is in Section 4 that we apply our techniques to the study of the NRT nonlinear Schr\" odinger equation, and obtain, amongst others, the results mentioned above. The last section is devoted to a summary of our results and a brief discussion of possible future work.


\section{Generalized Sundman transformations}

In this section we focus on power-type transformations involving both the function and its derivative.

\subsection{Definitions and basic properties}
In this section we analyze the transformation which we will call the generalized Sundman transformation of power type and its inverse. We begin with their definitions, and then discuss their properties and uses.

\begin{defi}[Direct power-type generalized Sundman transformation (GST)]\label{def:GSTs} 
Let $\alpha,\beta,a,b,\Lambda\in\mathbb{R}$, $\Lambda\neq 0$. We define the transformed function $Y(s)$ of a differentiable function $y=y(x)$ as

\begin{equation}\label{eq:GSTs1}
Y(s) =\,_{\Lambda}\mathcal S_{\alpha,\beta}^{a,b}[y](s)\equiv \,
y^\alpha(x)y'^\beta (x)
\end{equation}

where

\begin{equation}\label{eq:GSTs2}
s'(x)=\Lambda^{-1}y^a(x)y'^b(x),
\end{equation}
and where $'$ denotes a derivative with respect to $x$. $Y(s)=\,_{\Lambda}\mathcal S_{\alpha,\beta}^{a,b}[y](s)$, or $Y=\,_{\Lambda}\mathcal S_{\alpha,\beta}^{a,b}[y]$, will be called the \textit{direct generalized Sundman transformation (or transform)} of the function $y$.
\end{defi}

\begin{remark}\label{rem_GST_2}
	In the above definition $x=x(s)$ is obtained 
	by inverting the result of the integration $\int \Lambda^{-1}y^a(x)y'^b(x)\, dx=s+C$, 
	where $C$ is an arbitrary constant, and so $Y(s)$ is defined modulo a shift in its argument.
\end{remark}

\begin{remark}
    Note that if one lets $\beta=b=0$ in Definition~\ref{def:GSTs} the power-type Sundman transformation is trivially recovered. On the other hand, taking $\beta=0$ and $b\neq0$ we obtain a transformation which involves the first power of the derivative $y'$ in the equation governing the change of independent variable. Other types of GSTs with polynomials of first degree in $y'$ have, in fact, already been introduced in the study of differential equations \cite{CSL06,MR11}. However, as far as we know, the general case with four parameters $(\alpha,\beta,a,b)$ and $ \beta\neq0$ is employed here for the first time. The case with the parameter $\beta\neq0$ severely complicates the definition of the inverse transform, which we give below.
\end{remark}


\noindent
The definition of the inverse generalized Sundman transform is as follows: 

\begin{defi}[Inverse power-type generalized Sundman transformation (IGST)]\label{def:invGSTs}
Let $\alpha,\beta,a,b,\Lambda\in\mathbb R$, $\beta\neq 0$ and $\Lambda\neq0$. We define the inverse generalized Sundman transformation $_{\Lambda}\widetilde{\mathcal S}_{\alpha,\beta}^{a,b}$ of a function $Y=Y(s)$, denoted $_{\Lambda}\widetilde{\mathcal S}_{\alpha,\beta}^{a,b}[Y],$ or $_{\Lambda}\widetilde{\mathcal S}_{\alpha,\beta}^{a,b}[Y](x)$ when written as a function of the variable $x$, as

\begin{equation}\label{eq:invGSTs1}
y(x)=\,
_{\Lambda}\widetilde{\mathcal S}_{\alpha,\beta}^{a,b}[Y](x)
\,\equiv\,
\mathcal G(s),
\end{equation}

where, defining $\Omega$ via
\begin{equation}\label{eq:invGSTs3a}
\Omega+1=\frac{\beta(a+1)+\alpha(1-b)}{\beta},
\end{equation}
the function $\mathcal G(s)$ is defined as 

\begin{equation}\label{eq:invGSTs3}
\mathcal G(s)=\left(\Lambda(\Omega+1)\int Y(s)^{\frac{1-b}\beta }ds+K_1\right)^{\frac{1}{\Omega+1}},
\end{equation}

for $\Omega\neq-1$, or

\begin{equation}\label{eq:invGSTs4}
\mathcal G(s)=K_3\,\exp\left(\Lambda\int Y(s)^{\frac{1-b}\beta}ds\right),
\end{equation}

for $\Omega=-1$, and where

\begin{equation}\label{eq:invGSTs2}
x(s)=  K_2+\int \frac{\Lambda}{\mathcal G(s)^{a-\frac {\alpha b}\beta}Y(s)^{\frac b\beta}}\, ds,
\end{equation}

with $K_1$, $K_2$ and $K_3$ in the above formulae being arbitrary constants.

\end{defi}

\begin{remark}
		In this paper, our aim is to use the GST and its inverse in order to seek solutions $y(x)$ of differential equations by transforming first of all to new differential equations in $Y(s)$, and so the validity of the above transformations can only be checked a posteriori. This is not only because $y(x)$ and $Y(s)$, each assumed to be defined in some domain of $\mathbb R$ or $\mathbb C$, are not known at the outset, but also because their properties, and indeed their domains of definition, may depend on constants of integration as well as on any parameters in, and other aspects of, the differential equations they satisfy. For further comments on the use of GST as a means of solving differential equations, we refer to Remark~\ref{eqsols} below.
\end{remark}

\begin{remark}\label{kscal}
		In the above definition of the GST, the parameter $\Lambda$ incorporates a possible rescaling and/or reflection of the new independent variable $s$, which can be useful in order to avoid having to redefine $s$ after deriving the new differential equation in $Y(s)$.
				
\end{remark}


\begin{prop}\label{prop:inverse}
Let $\alpha,\beta,a,b,\Lambda\in\mathbb R$, $\beta\neq 0$  and $\Lambda\neq0$. Then, from Definitions~\ref{def:GSTs} and \ref{def:invGSTs}, it follows that

\begin{equation}
Y(s)=   \,_{\Lambda}\mathcal{S}_{\alpha,\beta}^{a,b}[y](s) \Longleftrightarrow y(x)= \,_{\Lambda}\widetilde{\mathcal S}_{\alpha,\beta}^{a,b}[Y](x)
\end{equation}
\end{prop}

\begin{proof}
	
We give a proof in each direction separately.
	
\subsubsection*{${\rm Part\ 1: \eqref{eq:GSTs1},\ \eqref{eq:GSTs2}} \Rightarrow {\rm \eqref{eq:invGSTs1},\ \eqref{eq:invGSTs3a},\ \eqref{eq:invGSTs3},\ \eqref{eq:invGSTs4},\  \eqref{eq:invGSTs2}}$}
We start with two functions $y(x)$ and $Y(s)$ such that $Y(s)=\,_{\Lambda}\GSTs[y](s)$, where $\,_{\Lambda}\GSTs$ is as given in Definition~\ref{def:GSTs}. 	
First, writing Eq.~\eqref{eq:GSTs2} as $s'(x)=\Lambda^{-1}y^ay'^{b-1}y',$ and then making use of Eq.~\eqref{eq:GSTs1}, it follows that

\begin{equation}
s'(x)
=\Lambda^{-1}y(x)^a\left(\frac{Y(s)}{y(x)^\alpha}\right)^{\frac{b-1}\beta} y'(x)
=\Lambda^{-1}y(x)^{\Omega}Y(s)^{\frac{b-1}\beta} y'(x),
\end{equation}
with $\Omega=a+\alpha\frac{1-b}\beta$. Then, by integrating, one finds

\begin{equation}
\int Y(s)^{\frac{1-b}\beta }ds
=\int \Lambda^{-1} [y(x)]^\Omega y'(x)dx
=\frac{[y(x)]^{\Omega+1}}{\Lambda(\Omega+1)}+K,\qquad \Omega \neq -1,
\end{equation}
for some arbitrary constant $K$.
Thus, defining
\begin{equation}
\mathcal F(s)=\int Y(s)^{\frac{1-b}\beta }ds,
\end{equation}
it follows that, for $\Omega\neq-1$,

\begin{equation} \label{eq:G(s)}
y(x)=\left(\Lambda(\Omega+1)\mathcal F(s)+K_1\right)^{\frac1{\Omega+1}},\qquad 
\end{equation}
where $K_1=-\Lambda(\Omega+1)K$. Meanwhile, in the case $\Omega=-1,$ one simply has

\begin{equation}\label{eq:G(s)b}
y(x)= K_3 e^{\Lambda\mathcal F(s)}
\end{equation}
for some arbitrary constant $K_3$.
For simplicity, we will take the following notation:

\begin{equation}\label{eq:G(s)c}
\mathcal G(s)=\left(\Lambda(\Omega+1)\mathcal F(s)+K_1\right)^{\frac1{\Omega+1}},\quad \Omega \neq -1,
\end{equation}
and 
\begin{equation}\label{eq:G(s)d}
\mathcal G(s)= K_3 e^{\Lambda\mathcal F(s)},\quad \Omega = -1.
\end{equation}
Summarizing, from Eqs.~\eqref{eq:G(s)},~\eqref{eq:G(s)b}, ~\eqref{eq:G(s)c} and~\eqref{eq:G(s)d}, one obtains that necessarily

\begin{equation}\label{eq:y=G}
y(x)=\mathcal G(s),
\end{equation}
as required.
With respect to the change of independent variable, from Eqs.~\eqref{eq:GSTs1} and~\eqref{eq:GSTs2} it directly follows that
	\begin{equation}
		s'(x)
		=\Lambda^{-1}y(x)^a\left(\frac{Y(s)}{y(x)^\alpha}\right)^{\frac b\beta}
		=\Lambda^{-1}y(x)^{a-\frac{\alpha b}{\beta}}Y(s)^{\frac b\beta}.
	\end{equation}
	Finally, by substituting Eq.~\eqref{eq:y=G} in the latter equation one finds
	\begin{equation}
		s'(x)=\Lambda^{-1}\mathcal G(s)^{a-\frac {\alpha b}\beta}Y(s)^{\frac b\beta},
	\end{equation}
	from where 
	\begin{equation}
x(s)=\int \frac{\Lambda}{\mathcal G(s)^{a-\frac {\alpha b}\beta}Y(s)^{\frac b\beta}}\, ds + K_2,
	\end{equation}
with $K_2$ being an arbitrary constant. This is just Eq.~\eqref{eq:invGSTs2}.

\subsubsection*{${\rm Part\ 2: \eqref{eq:GSTs1},\ \eqref{eq:GSTs2}} \Leftarrow {\rm \eqref{eq:invGSTs1},\ \eqref{eq:invGSTs3a},\ \eqref{eq:invGSTs3},\ \eqref{eq:invGSTs4},\  \eqref{eq:invGSTs2}}$}
Let us  now consider two functions  $Y(s)$ and $y(x)$ such that $y(x)=\,_{\Lambda}\widetilde{\mathcal S}_{\alpha,\beta}^{a,b}[Y](x),$ where $\,_{\Lambda}\GSTIs$ is as given in Definition~\ref{def:invGSTs}.

For $\Omega\neq-1$, Eqs.~\eqref{eq:invGSTs1}, ~\eqref{eq:invGSTs3a}, ~\eqref{eq:invGSTs3} and ~\eqref{eq:invGSTs2} yield
\begin{eqnarray}
y'(x) & = & \frac{1}{\Omega+1}\left(\Lambda(\Omega+1)\int Y(s)^{\frac{1-b}\beta }ds+K_1\right)^{-\frac{\Omega}{\Omega+1}}\Lambda(\Omega+1)Y(s)^{\frac{1-b}\beta }s'(x) \nonumber \\
& = & G(s)^{-\Omega}\Lambda Y(s)^{\frac{1-b}\beta }
\frac{1}{\Lambda}\mathcal G(s)^{a-\frac {\alpha b}\beta}Y(s)^{\frac b\beta}
\nonumber \\
& = & G(s)^{-a-\frac {\alpha (1-b)}\beta}G(s)^{a-\frac {\alpha b}\beta}Y(s)^{\frac{1}\beta }\nonumber \\
& = & G(s)^{-\frac {\alpha}\beta}Y(s)^{\frac{1}\beta }.
\end{eqnarray}
Meanwhile, for $\Omega=-1$, Eqs.~\eqref{eq:invGSTs1}, ~\eqref{eq:invGSTs3a}, ~\eqref{eq:invGSTs4} and ~\eqref{eq:invGSTs2} yield
\begin{eqnarray}
y'(x) & = & {K_3}\,\exp\left(\Lambda\int Y(s)^{\frac{1-b}\beta}ds\right) \Lambda Y(s)^{\frac{1-b}\beta}s'(x) \nonumber \\
& = & G(s) \Lambda Y(s)^{\frac{1-b}\beta}
\frac{1}{\Lambda}\mathcal G(s)^{a-\frac {\alpha b}\beta}Y(s)^{\frac b\beta} \nonumber \\
& = & G(s)^{1+a-\frac {\alpha b}\beta}Y(s)^{\frac{1}\beta } \nonumber \\
& = & G(s)^{-\frac {\alpha}\beta}Y(s)^{\frac{1}\beta }.
\end{eqnarray}
Thus, since $y(x)=G(s)$, we obtain, for any $\Omega$,
\begin{equation}\label{eq:G(s)e}
Y(s)=G(s)^\alpha y'^{\beta}(x)=y^{\alpha}(x) y'^{\beta}(x),
\end{equation}
which is just Eq.~\eqref{eq:GSTs1}. In order to derive Eq.~\eqref{eq:GSTs2},
note that Eq.~\eqref{eq:invGSTs2} gives
\begin{equation}
s'(x)=\frac{1}{\Lambda}\mathcal G(s)^{a-\frac {\alpha b}\beta}Y(s)^{\frac b\beta},
\end{equation}
and so, setting $G(s)=y(x)$ and making use of Eq.~\eqref{eq:G(s)e}, we obtain
\begin{equation}
s'(x)=\frac{1}{\Lambda} y(x)^{a-\frac {\alpha b}\beta}y(x)^{\frac {\alpha b}\beta}y'(x)^b
=\Lambda^{-1}y^a(x)y'^b(x),
\end{equation}
as required. This then completes the proof.

\end{proof}

\begin{remark}
We see from the above proposition that $\,_{\Lambda}\GSTs$ and $\,_{\Lambda}\GSTIs$ are indeed inverses. This can also be seen explicitly from the following Corollary, where we show that the composition $\,_{\Lambda}\GSTs\,_{\Lambda}\GSTIs$, which provides a mapping
$Y(s)\longrightarrow y(x)\longrightarrow \widetilde Y(\tilde s)$, is just the identity transformation (modulo translations in the independent variable).
\end{remark}

	\begin{coro}\label{coro:inverse}
		Let $\alpha,\beta,a,b,\Lambda\in\mathbb R$, $\beta\neq 0$ and $\Lambda\neq0$. Then the composition $\,_{\Lambda}\GSTs\,_{\Lambda}\GSTIs$, where $\,_{\Lambda}\GSTs$ and $\,_{\Lambda}\GSTIs$ are as defined in Definitions~\ref{def:GSTs} and \ref{def:invGSTs}, respectively, is the identity transformation (up to a translation).
	\end{coro}

	\begin{proof}
		
Let us consider the composition 
$\,_{\Lambda}\GSTs\,_{\Lambda}\GSTIs$, which provides a mapping $Y(s)\longrightarrow y(x)\longrightarrow \widetilde Y(\tilde s)$. By
definition, $Y$, $s$, $y$ and $x$ are related as in Eqs.~\eqref{eq:invGSTs1}, ~\eqref{eq:invGSTs3a}, ~\eqref{eq:invGSTs3}, ~\eqref{eq:invGSTs4} and ~\eqref{eq:invGSTs2}, and $y$, $x$, $\widetilde Y$ and $\tilde s$ satisfy
\begin{equation}\label{eq:GSTs1tila}
\widetilde Y(\tilde s) = y^\alpha(x)y'^\beta (x),\qquad 
\tilde s'(x)=\Lambda^{-1}y^a(x)y'^b(x).
\end{equation}
But, according to Proposition~\ref{prop:inverse}, a consequence of Eqs.~\eqref{eq:invGSTs1}, ~\eqref{eq:invGSTs3a}, ~\eqref{eq:invGSTs3}, ~\eqref{eq:invGSTs4} and ~\eqref{eq:invGSTs2} is that $Y$, $s$, $y$ and $x$
satisfy Eqs.~\eqref{eq:GSTs1} and ~\eqref{eq:GSTs2}, i.e.,
\begin{equation}\label{eq:GSTs1tilb}
Y(s) = y^\alpha(x)y'^\beta (x),\qquad s'(x)=\Lambda^{-1}y^a(x)y'^b(x).
\end{equation}
From ~\eqref{eq:GSTs1tila} and ~\eqref{eq:GSTs1tilb} we see that $\widetilde Y(\tilde s)=Y(s)$ and $s'(x)=\tilde s'(x)$, and so $\widetilde Y(\tilde s)=Y(\tilde s+C)$ for some
arbitray constant $C$. Since, as noted in Remark~\ref{rem_GST_2}, the GST $Y(s)$ of
a function $y$ is defined only up to shifts in its argument, this then implies that the composition $\,_{\Lambda}\GSTs\,_{\Lambda}\GSTIs$ corresponds to the identity mapping.
		
\end{proof}

\begin{remark}\label{eqsols}
		Our over-riding aim in this paper is to demonstrate the usefulness of the GST as a means of solving differential equations. Let us consider first of all the case of an autonomous second order differential equation in $y(x)$. It turns out that there are examples of such equations that the GST allows us to reduce to autonomous first order differential equations in the transformed variables $Y(s)$; this is in contrast to the use of Lie symmetries which, as is well-known, allows a reduction in the order of autnomous differential equations by one, but at the expense of transforming to nonautonomous equations. Then, in the case of a first order autonomous equation in the transformed variables, its solution $Y(s)=Y(s+K_0)$, where $K_0$ is an arbitrary constant, is to be used in Definition~\ref{def:invGSTs} in order to obtain the corresponding solution $y(x)$ of the second order equation under consideration. Here $K_0$ is redundant, as the combination $s+K_0$ is maintained throughout the process of determining $y(x)$; thus, in accordance with Remark~\ref{rem_GST_2}, we may set $K_0=0.$ We thus obtain from Definition~\ref{def:invGSTs} either a solution of the form $y(x)=f(x-K_2;K_1)$ in the case $\Omega\neq-1$, or of the form $y(x)=g(x-K_2;K_3)$ in the case $\Omega=-1$; here $K_1$ and $K_2$, or $K_3$ and $K_2$, are, respectively, the two required constants of integration.  
\end{remark}

\begin{remark}
Although we do not consider such problems in the present paper, the process for higher order autonomous equations is similar. For an autonomous equation of order $n\geq 3$ such that the GST may be used to obtain an autonomous equation in the transformed variables of order $n-1$, we may then use a solution $Y(s)=Y(s+K_0;L_1,\ldots,L_{n-2})$ of the latter depending on $n-1$ arbitrary constants, where once again we may set the redundant $K_0=0$, in order to obtain, using Definition~\ref{def:invGSTs}, a solution of the former of the form $y(x)=f(x-K_2;K_1,L_1,\ldots,L_{n-2})$ in the case $\Omega\neq-1$, or of the form $y(x)=g(x-K_2;K_3,L_1,\ldots,L_{n-2})$ in the case $\Omega=-1$. Here $K_2,L_1,\ldots,L_{n-2}$, along with either $K_1$ or $K_3$, are the required $n$ arbitrary constants of integration. Again we note the contrast with the use of Lie symmetries, whereby a transformation to an equation of order $n-1$, but nonautonomous, can be made. This last technique may, however, prove useful in solving autonomous equations in the transformed variables when the GST can be applied to autonomous differential equations of order $n\geq3$.
\end{remark}


\subsection{The subfamily of transformations with $\beta=0$}
In this section we analyze the subfamily of transformations $\beta=0$ and their relation with the non-linearity parameter of the travelling wave solutions of the  NRT nonlinear Schr\"odinger equation.  Letting $\beta=0$ in Eq.~\eqref{eq:GSTs1} gives
\begin{equation}\label{GST}
Y(s)=\, _{\Lambda}\mathcal S_{\alpha}^{a,b}[y](s) \equiv\, y^{\alpha}(x) ,\quad ds=\Lambda^{-1}\,y^{a}(x)y'^b(x)\,dx,
\end{equation}
where we adopt the notation $_{\Lambda}\mathcal S_{\alpha}^{a,b}=\ _\Lambda\mathcal S_{\alpha,0}^{a,b}$. We now show that this subfamily of transformations has the structure of a group provided that $\alpha\neq0$. We begin by showing that it is closed under composition.

\begin{lemma}[Composition law]\label{complaw}
The subfamily of transformations $\beta=0$ obeys the following composition law:
\begin{equation}\label{eq:comp_beta0}
_{\Delta}\mathcal S_{\delta}^{A,B}\circ\, _{\Lambda}\mathcal S_{\alpha}^{a,b}=\, _{j}\mathcal S_{\gamma}^{c,d},\quad  dz = j^{-1} y^c(x) (y')^d(x)\, dx
\end{equation}
with 
\begin{equation}\label{comp_par}
 \gamma = \alpha \delta,\quad c=a+A\alpha+ B(\alpha-a-1),\quad d=b+B(1-b), \quad \frac{1}{j} = \frac{\alpha^B}{\Lambda^{1-B}\Delta}\, .
\end{equation}
\end{lemma}
\begin{proof}
On the one hand, we have that
 \begin{equation*}
 	\left( \,_{\Delta}\mathcal S_{\delta}^{A,B}\circ\, _{\Lambda}\mathcal S_{\alpha}^{a,b} \right)[y](z)  = \left( \,_{\Delta}\mathcal S_{\delta}^{A,B}\left( \, _{\Lambda}\mathcal S_{\alpha}^{a,b} \right)[y](s) \right)(z) = \,_{\Delta}\mathcal S_{\delta}^{A,B} [Y(s)] (z) =  \widetilde Y(z),
 \end{equation*}
 where $Y(s) = y^{\alpha}(x)$, $s'(x) = \frac{1}{\Lambda}y^{a}(x)(y')^b(x)$, $\widetilde Y(z) = Y^{\delta}(s) = y^{\alpha\delta}(x)$ and $\dot{z}(s) = \frac{1}{\Delta} Y^{A}(s) (\dot{Y})^B(s)$, and where we use the notation $\dot{z} =\frac{dz}{ds}$, $\dot{Y} =\frac{dY}{ds}$.
 Taking into account that
 \begin{equation*}
 	\dot{z}(s) = z'(x)\frac{dx}{ds} = z'(x)\left(\frac{ds}{dx}\right)^{-1} \quad \text{and}\quad \dot{Y}(s) = \frac{d Y(s(x))}{dx}\frac{dx}{ds} =   \frac{d y^{\alpha}(x)}{dx}\left(\frac{ds}{dx}\right)^{-1} =  \alpha\, y^{\alpha-1}(x) y'(x)\left(\frac{ds}{dx}\right)^{-1},
 \end{equation*}
 we have that,
 \begin{align*}
 	z'(x)\frac{dx}{ds} &= \frac{1}{\Delta}(y^{\alpha}(x))^{A}\left[\frac{d}{dx}(y^{\alpha}(x))\frac{dx}{ds} \right]^{B} \\[0.5em]
 	z'(x) & = \frac{\alpha^{B}}{\Delta}y^{\alpha A + B(\alpha -1)} (x) [y'(x)]^{B} \Lambda^{B-1} y^{a(1-B)}(x)(y')^{b(1-B)}(x) \\
 	& = \frac{\alpha^{B}}{\Lambda^{1-B}\Delta} y^{a+A\alpha + B(\alpha-a-1)}(x)(y')^{b+B(1-b)}(x)\, .
 \end{align*}
 On the other hand, 
 \begin{equation*}
\left(  \,_{\Delta}\mathcal S_{\delta}^{A,B}\circ\, _{\Lambda}\mathcal S_{\alpha}^{a,b} \right)[y](z)  =  \,_{j}\mathcal S_{\gamma}^{c,d}[y](z) = \widetilde Y(z),
 \end{equation*}
 where $\widetilde Y(z) = y^{\gamma}(x)$ and $z'(x) = \frac{1}{j} y^{c}(x)(y')^{d}(x)$.  We thus obtain the composition rule
 \[
 \gamma = \alpha\delta, \quad c = a+A\alpha+B(\alpha-a-1), \quad d = b + B(1-b) \quad \text{and} \quad \frac{1}{j} = \frac{\alpha^{B}}{\Lambda^{1-B}\Delta }\, ,
 \]
as stated.

\end{proof}

\begin{coro}[Identity and inverse]\label{compid}First of all we note that the identity transformation is $\mathbb I=\, _{1}\mathcal S_{1}^{0,0}[y]$, since
	\[
	\, _{1}\mathcal S_{1}^{0,0}[y](s) = y^{1}(x) = y(x), \quad 
	\text{with} \quad ds = \frac{1}{1}y^0(x)(y')^0(x)dx = dx \,\,\, {\rm \, and\ so}\,\,  s = x\,  \text{(up to a constant; see Remark~\ref{rem_GST_2}}).
	\]
	Then, using Lemma~\ref{lemma:beta0_EDOs}, the inverse of the transformation $\, _{\Lambda}\mathcal S_{\alpha}^{a,b}$, denoted by $\left(\, _{\Lambda}\mathcal S_{\alpha}^{a,b}\right)^{-1}$, is obtained by requiring $\left(\, _{\Lambda}\mathcal S_{\alpha}^{a,b}\right)^{-1}=\, _{\Delta}\mathcal S_{\delta}^{A,B}$ to satisfy
	\begin{equation}\label{eq:comp_inv}
	\left(\, _{\Lambda}\mathcal S_{\alpha}^{a,b}\right)^{-1} \circ\, _{\Lambda}\mathcal S_{\alpha}^{a,b} =	\, _{\Delta}\mathcal S_{\delta}^{A,B}\circ\, _{\Lambda}\mathcal S_{\alpha}^{a,b}=\, _{1}\mathcal S_{1}^{0,0}=\mathbb I,\quad  
	\end{equation}
and so, setting $\gamma=1$, $c=0$, $d=0$ and $j=1$ in Eq.~\eqref{comp_par}, we have
	\begin{equation}
	\delta=\alpha^{-1},\quad A=\frac{a+(1-\alpha)b}{\alpha (b-1)},\quad  B=\frac b{b-1},\quad \Delta=\alpha^{\frac b{b-1}} \Lambda^{\frac1{b-1}}.
	\end{equation}
\end{coro}

\begin{coro}[Associativity]\label{compass}

Defining $\, _{\overline k}\mathcal S_{\overline \alpha}^{\overline a,\overline b}$ and $\, _{\widetilde k}\mathcal S_{\widetilde \alpha}^{\widetilde a,\widetilde b}$ by
\begin{equation}
\, _{\overline k}\mathcal S_{\overline \alpha}^{\overline a,\overline b}=\, _{k_3}\mathcal S_{\alpha_3}^{a_3,b_3}\circ\left(
\, _{k_2}\mathcal S_{\alpha_2}^{a_2,b_2}\circ\, _{k_1}\mathcal S_{\alpha_1}^{a_1,b_1}\right)
\end{equation}
and
\begin{equation}
\, _{\widetilde k}\mathcal S_{\widetilde \alpha}^{\widetilde a,\widetilde b}=\left(\, _{k_3}\mathcal S_{\alpha_3}^{a_3,b_3}\circ
\, _{k_2}\mathcal S_{\alpha_2}^{a_2,b_2}\right)\circ\, _{k_1}\mathcal S_{\alpha_1}^{a_1,b_1},
\end{equation}
it follows from the composition law given in Lemma~\ref{complaw} that
\begin{eqnarray}
\overline \alpha = \widetilde \alpha & = & \alpha_3\alpha_2\alpha_1,\\
\overline a = \widetilde a & = & a_1+a_2\alpha_1+a_3\alpha_2\alpha_1+b_3\alpha_1(\alpha_2-a_2-1)+(b_3+b_2-b_3b_2)(\alpha_1-a_1-1), \\
\overline b = \widetilde b & = & b_1+b_2(1-b_1)+b_3(1-b_2)(1-b_1), \\
(\overline k)^{-1} = (\widetilde k)^{-1} & = & 
\frac{\alpha_2^{b_3}\alpha_1^{b_3+b_2-b_3b_2}}{k_3k_2^{1-b_3}k_1^{(1-b_3)(1-b_2)}}.
\end{eqnarray}
We thus see that the composition defined in Lemma~\ref{complaw} is associative.
\end{coro}

From Lemma~\ref{complaw} and Corollaries~\ref{compid} and ~\ref{compass}, it follows that:

\begin{prop}
	The subfamily of transformations $\beta=0$ with $\alpha\neq0$ forms a group under composition.
\end{prop}

\begin{remark}
In contrast to the case $\beta\neq0$, the subfamily of transformations $\beta=0$  can always be applied to a given differential equation in $y(x)$. Moreover, it preserves the order of the equation, as we now see in the following example.
\end{remark}

\begin{lemma}\label{lemma:beta0_EDOs}
	Let $y=y(x)$ be a solution of the equation
	\begin{equation}\label{eq:EDO}
	(y')^l=Ay^m+By^n.
	\end{equation}  Then, the function $Y=Y(s)$ defined by equation \eqref{GST} with $\Lambda=\alpha^{-1}$ is a solution of the equation 
	\begin{equation}\label{eq:EDO'}
	(\dot{Y})^L=AY^M+BY^N,
	\end{equation} where  
	\begin{equation}\label{eq:lL}
	L=\frac l{1-b},\quad M=\frac m\alpha+\frac{l\,(\alpha-a-1)}{\alpha\,(1-b)},\quad N=\frac n\alpha+\frac{l\,(\alpha-a-1)}{\alpha\,(1-b)}.
	\end{equation}
\end{lemma}

\begin{proof}
Let $y$ be a solution of the equation $(y')^l=Ay^m+By^n.$ Then, starting form equation \eqref{GST} with $\Lambda=\alpha^{-1}$, we have
\begin{equation}\label{eq:2}
\dot Y(s)= y^{\alpha-a-1}(y')^{1-b}=y^{\alpha-a-1}(Ay^m+By^n)^{\frac{1-b}{l}}=Y^{\frac{\alpha-a-1}{\alpha}}(AY^{\frac m\alpha}+BY^{\frac n\alpha})^{\frac{1-b}{l}},
\end{equation}
or equivalently

\begin{equation*}
\dot Y^\frac{l}{1-b}= A\,Y^{\frac{l\,(\alpha-a-1)+m(1-b)}{\alpha\,(1-b)}}+B\,Y^{\frac{l\,(\alpha-a-1)+n(1-b)}{\alpha\,(1-b)}}.
\end{equation*}
That is to say
\begin{equation}\label{eq:3}
\dot Y^{L}= A\,Y^{M}+B\,Y^{N},
\end{equation}
with 
\begin{equation*}
L=\frac l{1-b},\quad M=\frac m\alpha+\frac{l\,(\alpha-a-1)}{\alpha\,(1-b)},\quad N=\frac n\alpha+\frac{l\,(\alpha-a-1)}{\alpha\,(1-b)}.
\end{equation*}
\end{proof}

\begin{remark} 
Let us observe the use of the parameter $\Lambda$ in Lemma~\ref{lemma:beta0_EDOs}, as noted in Remark~\ref{kscal}, the choice of which made here avoids having to redefine the new independent variable $s$ in order to preserve the coefficients $A$ and $B$.
\end{remark}

Taking into account that the transformations $\, _{\Lambda}\mathcal S_{\alpha}^{a,b}$ are composable and invertible, we see that the procedure in Lemma~\ref{lemma:beta0_EDOs} allows us to perform shifts in the family of equations~\eqref{eq:EDO}.\\

It is interesting to study how the generalized trigonometric functions, in the sense of Drábek and Manásevich~\cite{DRA99}, are affected by the group of transformations $_{\Lambda}\mathcal S_{\alpha}^{a,b}$, since these functions are involved in certain solutions of the NRT nonlinear Schrödinger equation, at least for the free particle case~\cite{GPPT25}.
\begin{coro}[Generalized trigonometric functions]
	The generalized sine and hyperbolic sine functions, $\sin_{l,n}$ and $\sinh_{l,n}$, satisfy Eq.~\eqref{eq:EDO} with $A=1,m=0$ and $B=-1$, and $A=1,m=0$ and $B=1$, i.e., 
	\begin{equation}\label{eq:pyth}
		(y')^l + y^n=1 \quad {\rm and}\quad (y')^l - y^n=1,
	\end{equation}
respectively. Applying Lemma~\ref{lemma:beta0_EDOs}, it follows that
	$$
	L=\frac{l}{1-b},\quad M=\frac{l(\alpha-a-1)}{\alpha(1-b)},\quad N=\frac n\alpha+\frac{l(\alpha-a-1)}{\alpha(1-b)}.
	$$
	If $a=\alpha-1$, one obtains
	$$
	L=\frac{l}{1-b},\quad M=0,\quad N=\frac n\alpha\, .
	$$
Thus, recalling that in  Lemma~\ref{lemma:beta0_EDOs} $\Lambda=\alpha^{-1}$,
	\begin{equation}
		\,_{\alpha^{-1}}\mathcal S_{\alpha}^{\alpha-1,b}[\sin_{l,n}]=\sin_{\frac{l}{1-b},\frac n\alpha},\qquad \,_{\alpha^{-1}}\mathcal S_{\alpha}^{\alpha-1,b}[\sinh_{l,n}]=\sinh_{\frac{l}{1-b},\frac n\alpha}.
	\end{equation}
	In addition, if $l=n$ and $b=1-\alpha$, these relations yield
	\begin{equation}
	\,_{\alpha^{-1}}\mathcal S_{\alpha}^{\alpha-1,1-\alpha}[\sin_{n}]=\sin_{\frac n\alpha},\qquad \,_{\alpha^{-1}}\mathcal S_{\alpha}^{\alpha-1,1-\alpha}[\sinh_{n}]=\sinh_{\frac n\alpha}.
	\end{equation}
	On the other hand, imposing $l=1$ in the equations~\eqref{eq:pyth}, the corresponding function $y$ is a generalized hyperbolic tangent function, $\tanh_n$, in the first case, and a generalized tangent function, $\tan_n$, in the second.
	Again setting $a=\alpha-1$, but this time taking also $b=0$, so that
	$$
	L=1,\quad M=0,\quad N=\frac n\alpha\, ,
	$$
	we obtain the relations
	\begin{equation}\label{tanhtan}
		\,_{\alpha^{-1}}\mathcal S_{\alpha}^{\alpha-1,0}[\tanh_{n}]=\tanh_{\frac n\alpha},\qquad \,_{\alpha^{-1}}\mathcal S_{\alpha}^{\alpha-1,0}[\tan_{n}]=\tan_{\frac n\alpha}.
	\end{equation}
\end{coro}

\begin{remark}
	We note that the transformation $\,_{\alpha^{-1}}\mathcal S_{\alpha}^{\alpha-1,0}$
	in Eq.~\eqref{tanhtan}, which has both $\beta=0$ and $b=0$, is in fact a standard Sundman transformation of power type.
\end{remark}

\section{On a class of autonomous second order differential equations}

In this section, we explore the connection between the power-like generalized Sundman transformations introduced in Definition~\ref{def:GSTs} of the previous section and the family of differential equations given by
$$
y''+\mathcal K\frac{y'^2}{y}=y^Ay'^Bf(y^\alpha y'^\beta).
$$
The results obtained here, in Proposition~\ref{prop:sols1} and Corollary~\ref{coro:EDO2} below, will then be used in the next section, in our discussion of solutions of the NRT nonlinear Schr\"odinger equation. In addition to the above equation with arbitrary function $f$, we will also consider here a particular example, which, in addition to providing an explicit illustration of the steps involved in deriving solutions via GSTs, also has the advantage that the results obtained are easily checked.

From now on, i.e., in this and the following section, the parameter $\Lambda$ will be set equal to unity. For the sake of readability, we will use the following notation: $_{1}\mathcal S_{\alpha,\beta}^{a,b}\equiv\mathcal S_{\alpha,\beta}^{a,b}$, and similarly $_{1}\widetilde{\mathcal S}_{\alpha,\beta}^{a,b}\equiv \widetilde{\mathcal  S}_{\alpha,\beta}^{a,b}.$
\begin{prop}\label{prop:sols1}
	Given $\alpha,\beta, A, B$ such that $\frac{\alpha}{\beta}=\frac{A+1}{B-2}$, then the solutions of the equation
	\begin{equation}\label{eq:EDO1}
	y''+\mathcal K\frac{y'^2}{y}=y^Ay'^Bf(y^\alpha y'^\beta)
	\end{equation}
	are given by
	\begin{equation}\label{eqprop:EDO1}
	y(x)= \widetilde {\mathcal S}_{\alpha,\beta}^{a,b}[Y](x),\qquad s=\int \frac{dY}{\beta Y^\mu f(Y)+( \alpha-\beta \mathcal K) Y^l}
	\end{equation}
	with
	\begin{equation}
	b=1+\frac{\beta(a+1)}{\alpha},\quad \mu=1+\frac{A-a}\alpha,\quad l=1-\frac{a+1}{\alpha}.
	\end{equation}

\end{prop}

\begin{proof}
	
	We start by considering two functions $y(x)$ and $Y(s)$ related through the generalized Sundman transform as given in Definition~\ref{def:GSTs}, $Y=\GSTs[y]$. We will denote $\dot Y(s)=\frac {d}{ds} Y(s)$, and $y'(x)=\frac {d}{dx} y(x).$ Then from Eqs.~\eqref{eq:GSTs1} and~\eqref{eq:GSTs2} follows
	\begin{equation}\label{eq:dotY}
	\dot Y(s)=\alpha y^{\alpha-1-a}y'^{\beta+1-b}+\beta y^{\alpha-a}y'^{\beta-1-b}y ''=\beta y^{\alpha-a}y'^{\beta-1-b}\left(y''+\frac{\alpha}{\beta}\frac{y'^2}{y}\right),
	\end{equation}
 and so
	\begin{equation}\label{eq:dotY2}
	\dot Y(s)+\left(\beta \mathcal K-\alpha \right) Y(s)^l=\beta y^{\alpha-a}y'^{\beta-1-b}\left(y''+\frac{\alpha}{\beta}\frac{y'^2}{y}\right) + \left(\beta \mathcal K-\alpha \right) y^{\alpha l} y'^{\beta l}.
	\end{equation}
	Now we select $l$ such that 
	\begin{equation}
	\begin{cases}
	\alpha l=\alpha-a-1,\\
	\beta l= \beta -b +1,
	\end{cases}
	\end{equation}
	which is possible only when 
	\begin{equation}\label{eq:cond1}
	\frac{a+1}\alpha=\frac{b-1}\beta,
	\end{equation}
	or equivalently $b=1+\frac{\beta(a+1)}{\alpha}$. We can then rewrite Eq.~\eqref{eq:dotY2} as 
	\begin{equation}\label{eq:dotY3}
	\dot Y(s)+\left(\beta \mathcal K-\alpha \right)Y(s)^l=\beta y^{\alpha-a}y'^{\beta-1-b}\left(y''+\mathcal K\frac{y'^2}{y}\right)
	\end{equation}
	Now, let us suppose that $y(x)$ is a function such that  $y''+\mathcal K\frac{y'^2}{y}=y^Ay'^Bf(y^\alpha y'^\beta).$ In this case, Eq.~\eqref{eq:dotY3} becomes
	\begin{equation}\label{eq:dotY4}
	\dot Y(s)+\left(\beta \mathcal K-\alpha \right)Y(s)^l= \beta y^{\alpha-a+A}y'^{\beta-1-b+B}f(y^\alpha y'^\beta),
	\end{equation}
	and we now select the parameters such that 
	\begin{equation}
	\alpha-a+A=\mu \alpha,\qquad\beta-1-b+B=\mu \beta, 
	\end{equation}
	for some $\mu\in\mathbb R$, which is possible only when
	\begin{equation}\label{eq:cond2}
	\frac{a-A}{\alpha}=\frac{1+b-B}\beta,
	\end{equation}
	or equivalently $b=\frac\beta\alpha(a-A)+B-1.$ Note that conditions in Eqs.~\eqref{eq:cond1} and~\eqref{eq:cond2} imply
	\begin{equation}
	\frac{\alpha}{\beta}=\frac{A+1}{B-2}.
	\end{equation}
	In this case, we can substitute Eq.~\eqref{eq:GSTs1} in Eq.~\eqref{eq:dotY4} to obtain
	\begin{equation}
	\dot Y(s)+\left(\beta \mathcal K-\alpha \right)Y(s)^l=\beta \left(y^{\alpha}y'^{\beta}\right)^\mu f(Y)=\beta Y(s)^\mu f(Y(s)),
	\end{equation}
	which is a first order autonomous equation in $Y$ and $s$, whose solution is given by
	\begin{equation}
	\int \frac{dY}{\beta Y^\mu f(Y)+( \alpha-\beta \mathcal K) Y^l}= s+K_0. 
	\end{equation}
	This solution $Y=Y(s+K_0)$ is the GST of the sought-after solution $y(x)$ of Eq.~\eqref{eq:EDO1}.
	Finally, bearing in mind that Eq.~\eqref{eq:EDO1} is an autonomous equation, we can set ${K_0}=0$  (see Remarks~\ref{rem_GST_2} and ~\ref{eqsols}), and then by Proposition~\ref{prop:inverse} we have that 
	\begin{equation}
	y(x)=\widetilde {\mathcal S}_{\alpha,\beta}^{a,b}[Y](x),
	\end{equation}
	which finishes the proof. 
\end{proof}

\begin{coro} \label{coro:EDO2}
Given $\gamma \in \mathbb{R}$ and an arbitrary function $g$, then the general solution of the family of autonomous second order differential equations 
	\begin{equation}\label{eq:EDO2}
	\frac{yy''}{y'^2}=g(y^\gamma\, y')
	\end{equation}
is given by 
	\begin{equation}\label{eq:EDO2_sol}
	y(x)=\widetilde {\mathcal S}_{\gamma\beta,\beta}^{a,b}[Y](x),\qquad s=\int\frac{Y^{\frac{a+1}{\beta\gamma}-1}}{\beta\gamma+ \beta g\left(Y^{1/\beta}\right)}dY ,
	\end{equation}
	with $b=1+\frac{a+1}{\gamma},$ and for arbitrary values of the parameters $a$ and $\beta$. Letting $\beta=1$ and choosing $a=\gamma-1$, we have that $\frac{a+1}{ \gamma}-1 =0$ which in turn implies $b=2$. This then allows a simplification of the result~\eqref{eq:EDO2_sol}, and the solution of Eq.~\eqref{eq:EDO2} can be expressed as
	\begin{equation}\label{eq:EDO2_simple}
y(x)= \widetilde {\mathcal S}^{\gamma-1,2}_{\gamma,1}[Y](x),\qquad s=\int\frac{dY}{\gamma+g(Y)}.
	\end{equation}
	\begin{proof}
On the one hand, Eq.~\eqref{eq:EDO1} can be written as 
		\begin{equation}
		\frac{y y ''}{y'^2}=(y^{\frac{A+1}{B-2}}y')^{B-2}\tilde f(y^\frac\alpha\beta  y')-\mathcal K,
		\end{equation}
		with $\tilde f(t)=f(t^\beta)$. Then, using $\frac{A+1}{B-2}=\frac{\alpha}{\beta}$, and setting $\alpha=\beta\gamma$,
		
		\begin{equation}
		\frac{y y ''}{y'^2}=(y^{\frac{\alpha}{\beta}}y')^{B-2}\tilde f(y^\frac\alpha\beta y')-\mathcal K=g(y^\gamma y'),
		\end{equation}
		with $g(t)=t^{B-2}\tilde f(t)-\mathcal K$, or equivalently $f(t)=t^{\frac{2-B}\beta}(g(t^\frac1\beta)+\mathcal K).$  We thus obtain Eq.~\eqref{eq:EDO2}.
		
 On the other hand, the solution Eq.~\eqref{eqprop:EDO1} of Eq.~\eqref{eq:EDO1},  with the function $f$ being  $f(t)=t^{\frac{2-B}\beta}(g(t^\frac1\beta)+\mathcal K)$ and using $l=1-\frac{a+1}{\alpha}$ and $\mu-l=\frac{A+1}{\alpha}=\frac{B-2}{\beta}$, gives 
\begin{equation}
y(x)= \widetilde {\mathcal S}_{\alpha,\beta}^{a,b}[Y](x),\qquad s=\int\frac{Y^{-l}}{\beta Y^{\mu-l}Y^{\frac{2-B}\beta}(g(Y^\frac1\beta)+\mathcal K) +\alpha-\beta{\mathcal K}}\, dY
=\int\frac{Y^{\frac{a+1}{\alpha}-1}}{\alpha+ \beta g\left(Y^{1/\beta}\right)}\, dY ,
\end{equation}
where $b=1+\frac{\beta(a+1)}{\alpha}$. Setting $\alpha=\beta\gamma$, as in the first part of the proof, then gives the results in~\eqref{eq:EDO2_sol}, where now $b=1+\frac{a+1}{\gamma}$ and $a$ and $\beta$ are left arbitrary. The simplification via $\beta=1$ and $a=\gamma-1$ in order to obtain the expression~\eqref{eq:EDO2_simple} for the solution is immediate. This then finishes the proof.
\end{proof}
\end{coro}

\begin{remark}
	The family of quasi-homogeneous equations 
$y'=\frac{y}{x} g(x^\gamma y)$
	is directly derived from Eq.~\eqref{eq:EDO2} by the substitution $y'=w(y)$ and relabelling $w\mapsto y$ and $y\mapsto x$. The change of variable $z(x)=x^\gamma y(x)$ then yields the separable equation $z'=\frac{z}{x}(\gamma+g(z))$ ~\cite[Section 1.7.1]{PZ03}, so giving an alternative approach to solving Eq.~\eqref{eq:EDO2} to that given above. In either approach, the problem of solving Eq.~\eqref{eq:EDO2} is reduced to a sequence of quadratures (and inverting one function).


\end{remark}

	\noindent{\bf A simple example.}
	
\noindent	
Let us consider the case $f=0$ of Eq.~\eqref{eq:EDO1}, i.e., 
\begin{equation}\label{sime}
y''+{\mathcal K}\frac{y'^2}{y}=0.
\end{equation} 
This equation is trivially solved, and has general solution $y(x)=(C_0x+C_1)^{\frac{1}{\mathcal K+1}}$ for $\mathcal K\neq-1$, where $C_0$ and $C_1$ are two arbitrary constants, or $y(x)=e^{D_0x+D_1}$ for $\mathcal K=-1$, where $D_0$ and $D_1$ are again two arbitrary constants. Let us now consider solving this equation via the GST, using the results given in Proposition~\ref{prop:sols1}.

First of all, the integral in Eq.~\eqref{eqprop:EDO1} gives the GST of $y(x)$:
\begin{equation}
Y(s)=\left(\frac{a+1}{\alpha}(\alpha-\beta\mathcal K)s\right)^\frac{\alpha}{a+1}.
\end{equation}
We must now use the inverse GST, as specified in Eq.~\eqref{eqprop:EDO1}, to obtain $y(x)$. The relation $b=1+\frac{\beta(a+1)}{\alpha}$ corresponds to the case $\Omega=-1$ of the inverse GST in Definition~\ref{def:invGSTs}, and so we have
\begin{eqnarray}\label{gint}
\mathcal G(s) & = & K_3\,\exp\left(\int Y(s)^{\frac{1-b}\beta}ds\right)
=K_3\,\exp\left(\int Y(s)^{-\frac{a+1}\alpha}ds\right) \nonumber \\
& = & K_3\,\exp\left(\int \left(\frac{a+1}{\alpha}(\alpha-\beta\mathcal K)s\right)^{-1}ds\right)=K_3\, s^{\frac{\alpha}{a+1}\frac{1}{\alpha-\beta\mathcal K}}.
\end{eqnarray}
Then
\begin{equation}
x-K_2=\int \frac{1}{\mathcal G(s)^{a-\frac {\alpha b}\beta}Y(s)^{\frac b\beta}}\, ds=\int \frac{1}{Es^p}\, ds
\end{equation}
where
\begin{equation}
E = K_3^{a-\frac{\alpha b}{\beta}}\left(\frac{a+1}{\alpha}(\alpha-\beta\mathcal K)\right)^{\frac{\alpha}{a+1}\frac{b}{\beta}}, \quad
p = \frac{\alpha}{a+1}\frac{1}{\beta}\left(\frac{a\beta-\alpha b}{\alpha-\beta \mathcal K}+b\right).
\end{equation}
Substituting $b=1+\frac{\beta(a+1)}{\alpha}$, we find
\begin{equation}
p=\frac{a(\alpha-\beta\mathcal K)-\mathcal K(\alpha+\beta)}{(a+1)(\alpha-\beta\mathcal K)}=1-\frac{\alpha(\mathcal K+1)}{(a+1)(\alpha-\beta\mathcal K)}.
\end{equation}
So
\begin{equation}\label{pint}
x-K_2=\int \frac{1}{Es^p}\, ds=\frac{s^{1-p}}{E(1-p)}
=\frac{(a+1)(\alpha-\beta\mathcal K)}{E\alpha(\mathcal K+1)}s^\frac{\alpha(\mathcal K+1)}{(a+1)(\alpha-\beta\mathcal K)}
\end{equation}
for $p\neq 1$, i.e., $\mathcal K\neq -1$, and, substituting from Eq.~\eqref{pint} into Eq.~\eqref{gint}, we obtain
\begin{equation}
y(x)=\mathcal G(s)=K_3\, s^{\frac{\alpha}{a+1}\frac{1}{\alpha-\beta\mathcal K}}
=K_3\left(\frac{E\alpha(\mathcal K+1)}{(a+1)(\alpha-\beta\mathcal K)}\right)^\frac{1}{\mathcal K+1}(x-K_2)^\frac{1}{\mathcal K+1}.
\end{equation}
This is precisely the general solution $y(x)=(C_0x+C_1)^{\frac{1}{\mathcal K+1}}$ for $\mathcal K\neq-1$, with $C_0$ being defined in terms of $K_3$ and $C_1$ in terms of both $K_2$ and $K_3$. On the other hand, for the case $p=1$, i.e., $\mathcal K=-1$, evaluating the first integral in Eq.~\eqref{pint} and inverting, we obtain $s=e^{E(x-K_2)}$ and so
\begin{equation}
y(x)=\mathcal G(s)=K_3\, s^{\frac{\alpha}{a+1}\frac{1}{\alpha-\beta\mathcal K}}=K_3\, e^{\frac{\alpha E}{(a+1)(\alpha+\beta)}(x-K_2)}.
\end{equation}
This is the general solution $y(x)=e^{D_0x+D_1}$ for $\mathcal K=-1$, with $D_0$ being defined in terms of $K_3$ and $D_1$ in terms of both $K_2$ and $K_3$.

Clearly, we do not advocate using GSTs in order to solve Eq.~\eqref{sime}, although the above calculations could be simplified by a suitable choice of parameters. Our aim in this example is to illustrate the steps to be taken in using GSTs to solve differential equations, moreover for an equation where the results are easily verified.

\section{Applications to the NRT nonlinear Schrödinger equation}

We now consider the application of the GST to Eq.~\eqref{eq:Psi_eq+travGen} with
$\Xi(\psi,\psi')$ as given in Eq.~\eqref{eq:Gamma_cond}.
From now on, $\psi(z)$ denotes the solution of the original second order differential equation Eq.~\eqref{eq:Psi_eq+travGen} which we aim to solve, and $\Psi(s)$ denotes the solution of the transformed ODE under the GST. 
Let us note that the structure of  $\Xi(\psi,\psi')$ as given in Eq.~\eqref{eq:Gamma_cond} includes such relevant cases as the zero potential $\Xi(\psi,\psi')=0$ and $\Xi(\psi,\psi')=\lambda\psi^q.$ One can easily verify that these are the only cases in which the term $\Xi(\psi,\psi')$ of the form~\eqref{eq:Gamma_cond} depends only on $\psi$.
Let us also recall, as mentioned earlier, that in this section we take $\Lambda=1$.
We first of all consider the case \eqref{eq:Gamma_cond} in general, and then turn to the study of various subcases.

\subsection{The general case $\Xi(\psi,\psi')=\hbar \omega \psi'(z)-\frac{(\psi')^2}{\psi^{q}}\,\overline g\left(\psi^\gamma \psi'\right)$}

\begin{teorema}\label{Th:solutionsNRT}
	The solution of Eq.~\eqref{eq:Psi_eq+travGen}
	\begin{equation*}
	\hbar \omega \psi'(z) = \frac{\hbar^2k^2}{2m}\left[\psi^{1-q}\psi''(z)+(1-q)\psi^{-q}(z)(\psi'(z))^2\right]+\Xi(\psi, \psi')\, 
	\end{equation*} 
	under the condition~\eqref{eq:Gamma_cond}
	\begin{equation*}
	\Xi(\psi, \psi')=\hbar \omega \psi'(z)-\frac{(\psi')^2}{\psi^{q}}\,\overline g\left(\psi^\gamma \psi'\right)	
	\end{equation*} 
	is given by
	\begin{equation}\label{eq:gen_sol}
	\psi(z)=\widetilde{\mathcal S}_{\beta\gamma,\,\beta}^{a,b}[\Psi],\qquad \beta s=
	\int\frac{\Psi^{\frac{a+1-\beta\gamma}{\beta\gamma}}}{\gamma+q-1+\frac{2m}{(\hbar k)^2} \overline g(\Psi^{1/\beta})} \, d\Psi,
	\end{equation}
	with $a\in\mathbb R$ and $b=1+\frac{a+1}{\gamma}$. Letting $\beta=1$ and $a=\gamma-1$ this last equation reads
	\begin{equation}\label{eq:par_sol}
	\psi(z)=\widetilde{\mathcal S}_{\gamma,\,1}^{\gamma-1,\,2}[\Psi],\qquad s =
	\int\frac{1}{\gamma+q-1+\frac{2m}{(\hbar k)^2} \overline g(\Psi)}\,d\Psi.
	\end{equation}
\end{teorema}
\begin{proof}
	Eq.~\eqref{eq:Psi_eq+travGen} can be written as
	\begin{equation}
	\frac{\psi\psi''}{(\psi')^2}=\frac{2m\omega}{\hbar k^2}\,\frac{\psi^q}{\psi'}\,\left(1-\frac{1}{\hbar\omega}\,\frac{\Xi(\psi,\psi')}{\psi'}\right)+q-1.
	\end{equation}
	Now, we set
	\begin{equation}\label{eq:EQ1}
	\frac{2m\omega}{\hbar k^2}\,\frac{\psi^q}{\psi'}\,\left(1-\frac{1}{\hbar\omega}\,\frac{\Xi(\psi,\psi')}{\psi'}\right)+q-1=g(\psi^\gamma\psi'),
	\end{equation}
	and then, comparing to Eq.~\eqref{eq:EDO2} in  Corollary~\ref{coro:EDO2}, one can write the solution of Eq.~\eqref{eq:Psi_eq+travGen} as a function of $\gamma$ and $g$ using Eq.~\eqref{eq:EDO2_sol}. Now, since
	\begin{equation}\label{eq:cond_pot}
	\Xi(\psi,\psi')=\hbar\omega \psi'-\frac{(\psi')^2}{\psi^q} \overline g(\psi^\gamma\psi'),
	\end{equation}
	 from Eq.~\eqref{eq:EQ1} one has
	
    \begin{equation}\label{eq:gg_funcs}
		g(\psi^\gamma\psi')=\frac{2m}{(\hbar k)^2}\overline g(\psi^\gamma\psi')+q-1.
		\end{equation}
	Finally, substituting Eq.~\eqref{eq:gg_funcs} in Eq.~\eqref{eq:EDO2_sol}, the solution of Eq.~\eqref{eq:Psi_eq+travGen} is obtained as
	\begin{equation}
	\psi(z)=\widetilde{\mathcal S}_{\beta\gamma,\,\beta}^{a,b}[\Psi],\qquad \beta s=
	\int\frac{\Psi^{\frac{a+1-\beta\gamma}{\beta\gamma}}}{\gamma+q-1+\frac{2m}{(\hbar k)^2 } \overline g(\Psi^{1/\beta})} \, d\Psi,
	\end{equation}
    where $b=1+\frac{a+1}{\gamma}$ as stated in Corollary~\ref{coro:EDO2}.	
	Setting $a=\gamma-1$ and $\beta =1$, which results in $b=2$, one obtains a simplified the expression for the solution given by
	\begin{equation}\label{eq:SOLCASEI}
	\psi(z) =\widetilde{\mathcal S}_{\gamma,\,1}^{\gamma-1,\,2}[\Psi],\qquad s=
	\int\frac{d\Psi}{\gamma+q-1+\frac{2m}{(\hbar k)^2} \overline g(\Psi)} \, ,
	\end{equation}
	corresponding to Eq.~\eqref{eq:EDO2_simple} in Corollary~\ref{coro:EDO2}, and thus finishing the proof. 
\end{proof}

\begin{remark}
    It is worth mentioning that for the choice $\gamma=C-q$ the function $\Psi=\Psi(s)$ does not depend on the non-linearity parameter $q$, which means that we can encapsulate the whole non-linearity behavior in the IGST. In particular, letting $\gamma = 1 -q$ and  denoting by $\psi_q$ the solution of Eq.~\eqref{eq:Psi_eq+travGen} under the condition~\eqref{eq:Gamma_cond} for an arbitrary $q$, one has that
    \begin{equation}\label{eq:gen_sol_simply}
    \psi_q(z) =\widetilde{\mathcal S}_{1-q,\,1}^{-q,\,2}[\Psi],\qquad s=\frac{(\hbar k)^2}{2m} 
	\int\frac{d\Psi}{\overline g(\Psi)}.
    \end{equation}
   As a direct consequence of this equation and Proposition~\ref{prop:inverse}, one finds that 
 \begin{equation}
    \Psi(s)={\mathcal S}_{1-q,\,1}^{-q,\,2}[\psi_q]={\mathcal S}_{1-\overline q,\,1}^{-\overline q,\,2}[\psi_{\overline q}]
 \end{equation}
   for any $q$ and $\overline q$, and so 
 \begin{equation}
    \psi_{\overline q}
    =\left(\widetilde{\mathcal S}_{1-\overline q,\,1}^{-\overline q,\,2}\circ{\mathcal S}_{1-q,\,1}^{-q,\,2}\right)[\psi_q].
 \end{equation}
\end{remark}
\begin{coro}\label{cor6}
Now, with the aim of obtaining the explicit solution of Eq.~\eqref{eq:Psi_eq+travGen} under the condition in Eq.~\eqref{eq:Gamma_cond}, we apply, as stated in Theorem~\ref{Th:solutionsNRT}, the IGST to the function $\Psi(s)$ given implicitly in Eq.~\eqref{eq:par_sol}. In particular, substituting the parameters of the IGST in Eq.~\eqref{eq:par_sol}  into~\eqref{eq:invGSTs3a} we find that
	\[
	\Omega = a+\alpha\frac{1-b}\beta = (\gamma -1) + \gamma(1-2) = -1,
	\]
    and so 
    \[\psi(z)=\mathcal G(s)=K_3\exp\left(\int\frac{ds}{\Psi(s)}\right),\qquad z=K_2+ \int\frac{\mathcal G(s)^{1+\gamma}}{\Psi(s)^2}ds.\]
    As in the above remark, we highlight the case $\gamma=1-q$, for which $\Psi(s)$ does not depend on the parameter $q$, to find 
      \[\psi(z)=\mathcal G(s)=K_3\exp\left(\int\frac{ds}{\Psi(s)}\right),\qquad z= K_2+\int\frac{\mathcal G(s)^{2-q}}{\Psi(s)^2}ds \, .\]
      \end{coro}


\subsection{Subcase I: $\ \overline g(\psi^\gamma\psi')=\lambda(\psi^{1-q}\psi')^\xi$}

As an example, we focus on the particular case $\overline g(\psi^\gamma\psi')=\lambda(\psi^{1-q}\psi')^\xi$ with $\xi\in\mathbb{R}$, where we note $\gamma=1-q$: 
\begin{equation}\label{eq:cond_caseII}
\Xi(\psi, \psi')=\hbar \omega \psi'(z)-\lambda\frac{(\psi')^{2+\xi}}{\psi^{q(1+\xi)-\xi}}.
\end{equation}
Using Eq.~\eqref{eq:gen_sol_simply} we obtain
\begin{equation}
 s
=\frac{(\hbar k)^2}{2m\lambda}\int \Psi^{-\xi}d\Psi
=\frac{(\hbar k)^2}{2m\lambda(1-\xi)}\Psi(s)^{1-\xi},
\end{equation}
from where one can explicitly express $\Psi$ in terms of $s$,
\begin{equation}\label{eq:111}
    \Psi(s)=\left(\frac{2m\lambda(1-\xi)}{(\hbar k)^2}s\right)^{\frac1{1-\xi}}=\kappa_1 s^{\frac{1}{1-\xi}},
\end{equation}
with $\kappa_1=\left(\frac{2m\lambda(1-\xi)}{(\hbar k)^2}\right)^{\frac1{1-\xi}}.$
In order to obtain the explicit solution of Eq.~\eqref{eq:Psi_eq+travGen} under the condition in Eq.~\eqref{eq:cond_caseII}, we apply the IGST $\Psi(s)$, as stated in Theorem~\ref{Th:solutionsNRT}. Using Corollary~\ref{cor6}, recalling that $\Omega=-1$ and, moreover, $\gamma=1-q$, we obtain that $\mathcal G(s)$ is given by
\begin{equation}
\mathcal G(s)=K_3\exp\left(\int\frac{ds}{\Psi(s)}\right)=K_3\exp\left(\kappa_2\,s^{\xi/(\xi-1)}\right),
\end{equation} 
with $\kappa_2 = \frac{(\xi-1) }{\xi \kappa_1}$.  To continue, again using Corollary~\ref{cor6}, we have
\cite[Sec. 2.33 (10)]{Gradshteyn2014}
\begin{align}\label{eq:inversexs}
z(s)
& =K_2 + \int \frac{\mathcal G(s)^{2-q}}{\Psi(s)^2}\,ds
= K_2+\frac{K_3^{2-q}}{\kappa_1^2}\int s^{\frac{2}{\xi-1}} \exp\left((2-q)\kappa_2 s^{\xi/(\xi-1)}\right)ds
\nonumber \\
&= K_2 + \frac{K_3^{2-q}}{\kappa_1^2}\frac{1-\xi}{\xi\,\left((q-2)\kappa_2\right)^{\frac{\xi+1}{\xi}}}\,\Gamma\left(\frac{\xi+1}{\xi},(q-2)\kappa_2\,s^{\xi/(\xi-1)}\right)
\nonumber \\
&=K_2 + \kappa_3^{-1}\,\Gamma\left(\frac{\xi+1}{\xi},(q-2)\kappa_2\,s^{\xi/(\xi-1)}\right)
\end{align}
where $\kappa_3=\frac{\kappa_1^2\xi\,\left((q-2)\kappa_2\right)^{\frac{\xi+1}{\xi}}K_3^{q-2 }}{(1-\xi)}.$ From Eq.~\eqref{eq:inversexs}, we can easily compute $s(z)$ as
\begin{equation}
s(z)=\left(\frac{\Gamma^{-1}\left(\frac{\xi+1}{\xi},\kappa_3(z-K_2)\right)}{(q-2)\kappa_2}\right)^{\frac{\xi-1}{\xi}}. 
\end{equation}
Finally, the general solution, $\psi(z)=\mathcal{G}(s)$ is given by
\begin{equation}
\psi(z)= K_3 \exp\left(\frac{\Gamma^{-1}\left(\frac{\xi+1}{\xi},\kappa_3(z-K_2)\right)}{q-2}\right), \quad \xi\in\mathbb{R}\, .
\end{equation}

\subsection{Subcase II: $\Xi(\psi)=0$ }
The zero potential case has been analyzed by the authors in~\cite{GPPT25}. In this section, we show that the transformations defined here perform a shift in the nonlinearity parameter, $q$, when they are applied to the wavefunction of  the traveling wave solution of Eq.~\eqref{eq:Psi_eq+travGen} with $\Xi(\psi,\psi')=0$. 
\begin{coro}[Traveling wave solutions]
	The traveling wave solution of the NRT nonlinear Schrödinger equation satisfies~\cite[Eq.(11)]{BN16}
	\begin{equation}\label{eq:freeNRT}
	\psi'(z)=k_1\psi^q(z)+k_2\psi^{q-1}(z),
	\end{equation}
	where $k_1 =\frac{2m\omega }{\hbar k^2}$ and $k_2$ depends on $k_1$ (Eq. (12) in~\cite{BN16}). This equation  is a particular case of equation~\eqref{eq:EDO} with $l=1,m=q,n=q-1.$ If we select $b=0$ in Lemma 2, then 
	$$
	L=1,\quad M=\frac{q+\alpha-a-1}{\alpha},\quad N=\frac{q-1+\alpha-a-1}{\alpha}=M-\frac1\alpha.
	$$
	Imposing $\alpha=1$ we obtain
	$$
	L=1,\quad M=q-a = \widetilde q,\quad N=(q-a)-1 = \widetilde q -1, 
	$$ 
		corresponding to the equation,
		\begin{equation}\label{eq:121}
		\dot\Psi(s)=k_1 \Psi(s)^{\widetilde q}+k_2 \Psi(s)^{\widetilde{q}-1}.
		\end{equation}

	Summarizing, the selection $\alpha=1, b=0$ in Lemma~\ref{lemma:beta0_EDOs} implies a shift in the non-linearity parameter of the traveling wave solutions from $q$ to $\widetilde{q}=q-a$, i.e.,
	\begin{equation}\label{eq:GST_travel_sols}
	\mathcal S_{1}^{a,0}[\psi_q]= \Psi_{\widetilde q}(s),\quad\text{or equivalently},\quad  \mathcal S_{1}^{a,0}[\psi_{q+a}]= \Psi_{q}(s),
	\end{equation}
where $\psi_q$ denotes the solution of Eq.~\eqref{eq:freeNRT}, and similarly $\Psi_{\tilde q}$ of~\eqref{eq:121}, 	and so
	\begin{equation}
	\mathcal S_{1}^{1-q,0}[\psi_1 ] = \Psi_{q}(s)
	\end{equation}
	where we have chosen $a=1-q$, and $\psi_1(z)$ is the solution of the first order linear equation $\psi' = k_1 \psi + k_2$, i.e.,
		\[
		\psi_1 (z) = K e^{k_1 z} - \frac{k_2}{k_1}\, , 
		\]
		with $K$ being an integration constant. That is, $\Psi_{q}$ may be expressed as a transformation of $\psi_1$.
	
\end{coro}

\begin{coro}\label{coro:8}
Let $W_p[y]$ denote the functional $W_p[y] = \int y(z)^p\, dz$, over a suitable domain, and let $\Psi_q(s)$ and $\psi_{q+a}(z)$ be two traveling wave solutions of Eq.~\eqref{eq:freeNRT} related as shown in~\eqref{eq:GST_travel_sols} through the GST given in~\eqref{GST}. Then the following equality is satisfied:	
		\begin{equation}\label{eq:normp1}
		W_p [\Psi_q] = W_{\tilde{p}}[\psi_{\tilde{q}}], \quad \tilde{p} = p+a, \,\,\, \tilde q = q+a \, .
		\end{equation}
		In particular, for $a=1-q$,
		\begin{equation}\label{eq:normap2}
		W_p [\Psi_q] = W_{\tilde{p}}[\psi_{1}], \quad \tilde{p} = p-q+1  \, .
		\end{equation}	
\end{coro}

	\begin{proof}
		We will prove the general case, i.e., Eq.~\eqref{eq:normp1}:
		\begin{equation*}
		W_p [\Psi_q] = \int \Psi_q(s)^p\, ds =  
		\int \psi_{q+a}(z)^p \psi_{q+a}(z)^{a} \, dz =
		\int \psi_{q+a}(z)^{p+a} \, dz = W_{p+a} [\psi_{q+a}] \, ,
		\end{equation*}
where we have used the two equations of the transformation $S_{1}^{a,0}$, i.e.,
${\Psi_q}(s)=\psi_{q+a}(z)$ and $ds=\psi_{q+a}(z)^{a} \, dz$.
	\end{proof}

\subsection{Subcase III: $\Xi(\psi)=\lambda \psi^q$}
Rego-Monteiro obtains the following equation for the wavefunction of the nonlinear NRT-Schrödinger equation under the action of a constant potential, $V(z)=\lambda$ with $\lambda\neq 0$ (see~\cite[Eq.(10)]{REG20b}), corresponding to Eq.~\eqref{eq:Psi_eq+travGen} with $\Xi(\psi)=\lambda \psi^q$:
\begin{equation}\label{eq:RM_loren}
\psi'=K_4\psi^{1-q}\left(\psi ''+(1-q)\frac{\psi'^2}{\psi}\right)+K_5\psi^q
\end{equation}
where $K_4 = \frac{\hbar k^2}{2m\omega}$ and $K_5 = \frac{\lambda}{\hbar \omega}$. Eq.~\eqref{eq:RM_loren} may be written
\begin{equation}
\psi ''+(1-q)\frac{\psi'^2}{\psi}
=\psi^{q-1}\psi' \left(\frac{1}{K_4}-\frac{K_5}{K_4}\frac{\psi^q}{\psi'}\right),
\quad
\end{equation}
which has the same structure as Eq.~\eqref{eq:EDO1} in Proposition~\ref{prop:sols1} with
\begin{equation}
\mathcal K=1-q, \quad
A=q-1,\quad B=1,\quad \alpha=q,\quad \beta=-1\, .
\end{equation}
In addition, the necessary condition, $\frac \alpha\beta=\frac{A+1}{B-2},$ in Proposition~\ref{prop:sols1} is fulfilled for any value of the non-linearity parameter $q$ since
\begin{equation}
\frac \alpha\beta=-q=\frac{A+1}{B-2}.
\end{equation}
Then, by virtue of Proposition \ref{prop:sols1} we have that  


\begin{equation}\label{eq:RM_sol1}
\psi(z)= \widetilde {\mathcal S}_{q,\,-1}^{a,b}\left[\Psi \right],\qquad s =\int \frac{K_4}{(K_5\Psi^2-\Psi+K_4)\Psi^{b}}\,d\Psi,
\end{equation}
with
\begin{equation}
b=1-\frac{a+1}{q} \, .
\end{equation}

To express the solution in a simpler way we select $a=q-1$, so $b=0$. 
In this way, the integral in Eq.~\eqref{eq:RM_sol1} can be easily evaluated:
\[
s =\int \frac{K_4}{(K_5\Psi^2-\Psi+K_4)}\,d\Psi=\frac{2 K_4}{\kappa} \arctan\left(\frac{2K_5 \Psi - 1}{\kappa}\right),\qquad \kappa=\sqrt{4 K_4K_5-1}.\]
Inverting, one obtains 
\begin{equation}
\Psi(s)=\frac1{2K_5}\left(\kappa \tan\left(\frac{\kappa\,s}{2K_4}\right)+1\right).
\end{equation} 
Now, to apply the inverse generalized Sundman transform, $\GSTIs$, we first note that $\beta(a+1)+\alpha(1-b)=0$ and so $\Omega =-1$, which implies that the function $\mathcal G$ involved in Definition~\ref{def:invGSTs} takes the form 
\begin{equation}
\mathcal G(s)=K_3\exp\left(\int \frac{ds}{\Psi(s)}\right),
\end{equation} 
where $K_3$ is an integration constant. One thus obtains
\begin{equation}\label{eq:class_sol}
\mathcal G(s)
=K_3 e^{\frac{1}{2K_4}s}\left(\cos\left(\frac{\kappa  s}{2 K_4}\right) + \kappa \sin\left(\frac{\kappa  s}{2 K_4}\right)\right).
\end{equation}

Further computations provide the following expression for the change of variables:
\begin{align}
z(s) & 
= K_2 + 2 K_4 K_3^{1-q} \frac{c_{\kappa}^{1-q}}{\kappa} e^{-\frac{1-q}{\kappa}\phi_{\kappa} }\int e^{\frac{1-q}{\kappa}t}[\cos(t)]^{1-q} \, dt  \, ,\label{eq:psiz3}
\end{align}
where $t=\frac{\kappa}{2K_4}s+\phi_{\kappa}$, $c_{\kappa} = \sqrt{1+\kappa^2}$ and $\phi_{\kappa} = -\arctan(\kappa)$.

\begin{remark}
    It is worth highlighting that, since both the direct and inverse GSTs depend on the non-linearity parameter $q$ because of the choice $a=q-1$ above, this dependence is lost for the functions $\Psi(s)$ and $\mathcal G(s)$, and is then recovered later by the solution $\psi(z)$. In this sense, we observe that  this choice of the parameters means that the dependence on the non-linearity parameter $q$ can be fully encapsulated in the generalized Sundman transformation, as expressed in Eq.~\eqref{eq:RM_sol1}.
    As a direct consequence,
    \begin{equation}
        \Psi(s)={\mathcal S}_{q,\,-1}^{q-1,0}[\psi_q(z)]={\mathcal S}_{\overline q,\,-1}^{\overline q-1,0}[\psi_{\overline q}(z)]
    \end{equation}
    for any $q$ and $\overline q$, and so from Proposition~\ref{prop:inverse} it follows that
\begin{equation}
\psi_{\overline q}(z)
=\left(\widetilde{\mathcal S}_{\overline q,\,-1}^{\ \overline q-1,0}\circ{\mathcal S}_{q,\,-1}^{q-1,0}\right)[\psi_q(z)],
    \end{equation}
    providing a shift transformation in the non-linearity parameter $q$ for the solutions.
\end{remark}

\begin{remark}
With respect to the NRT nonlinear Schrödinger equation, the solution of the free-particle potential cannot be obtained as a limiting case ($\lambda\to 0)$ of the constant potential, $\Xi(\psi) = \lambda \psi^q$, which means that they have to be considered as two different scenarios. This interesting behavior differs from the classical Schrödinger equation, in which a non-zero constant potential, $V_0=\lambda$, is equivalent to the free-particle potential since the former can be considered as a shift in the energy of the system.
\end{remark}


\begin{remark}
The solution of Eq.~\eqref{eq:RM_loren}, $\psi_q(z)$, can be expressed as a deformation of the classical solution, i.e., for $q=1$. More precisely, from Eq.~\eqref{eq:RM_sol1} and Definition~\ref{def:invGSTs} one has that the general solution for Eq.~\eqref{eq:RM_loren} is given by 
\begin{align}
\psi_q(z) & =\mathcal G(s),\\[0.5em]
z(s) & = K_2 + \int [\mathcal G(s)]^{1-q}\,ds \, .
\end{align}
Setting $q=1$, we obtain $\psi_1(z)={\mathcal G}(s)$ and $s=z$ (up to a constant), so
$G(s)$ is precisely the solution for $q=1$. The above two equations are therefore
equivalent to
\begin{align}
\psi_q(z) & =\mathcal \psi_1(s),\label{eq:psiz} \\[0.5em]
dz & =[\psi_1(s)]^{1-q}\,ds \, ,\label{eq:psiz2}
\end{align}
which corresponds to 
\begin{equation}\label{eq:440}
\psi_q(z)={\mathcal S}_{1}^{1-q,0}[\psi_1(s)].
\end{equation}
Thus the  nonlinear solution, $\psi_q(z)$, is related to the linear solution, $\psi_1(s)$, through  the subfamily of generalized Sundman transform of power-type defined in Eq.~\eqref{GST} with parameters $\alpha=1, a=1-q$ and $b=0$ (we note that, since $b=0$, this is in fact a Sundman transformation).
From Eq.~\eqref{eq:440} and the composition property in Eq.~\eqref{eq:comp_beta0} one easily derives
\begin{equation}
\mathcal S_1^{1-\overline q,0}[\psi_q(z)]=\psi_{q+\overline q-1}(s).
\end{equation}
\end{remark}
%
%

\begin{remark} There exists the following relation between $W_p[\psi_q]$ and  $W_{\tilde{p}}[\mathcal{G}]$ (see Corollary~\ref{coro:8}):
	$$
	W_p[\psi_q] = \int [\psi_q(z)]^p\, dz=\int \mathcal G(s(z))^p\,dz=
	\int [\mathcal G(s(z))]^{p+1-q}[\mathcal G(s(z))]^{q-1}\,dz=
	\int [\mathcal G(s)]^{p+1-q}\, ds  =W_{p+1-q}[\mathcal{G}]\, . 
	$$
	Thus,
	\begin{equation}
	W_{p}[\psi_q]=\ W_{\tilde{p}}[\mathcal{G}]   \, , \quad \tilde{p} = p+1-q\, . 
	\end{equation}
\end{remark}

\section{Conclusions and open problems}

In this paper we have defined a generalization of the power-type Sundman transformation, involving not only powers of the function but also of its derivative, along with an inverse transformation. These transformations also constitute a generalization of transformations, called birametric up/down transformations, recently introduced by one of the authors in the framework of information theory and which have proved to be a useful tool in generalizing sharp informational inequalities. Furthermore, we have employed these
new transformations, now depending on four parameters, in a totally different context to that considered previously: the study of ordinary differential equations. A particularly interesting feature of these transformations is that, for certain classes of autonomous equation of order $n\geq2$, the resulting transformed equations are also autonomous but of order $n-1$. In this paper we have seen how they may be used to reduce a certain family of second order autonomous nonlinear ordinary differential equations to first order autonomous ordinary differential equations. Applications to higher order equations will be considered in future papers.
	
As an application of these new transformations, we note that the family of equations considered in this paper includes equations of interest in the framework of the NRT nonlinear Schr\"odinger equation. This fact allows us to not only find an explicit expression for the general solution of some reductions recently studied by Rego-Monteiro with Lorentzian form in the case of a constant potential, but also to encapsulate the whole nonlinear behavior of the solutions in the transformations introduced in this paper. Specifically, we can write the solution as a transformed function of the solution of the linear case. In addition, the composition of transformations has been seen to perform a shift in the nonlinearity parameter $q$. Moreover, as an interesting byproduct, we have shown that a certain subfamily of these new transformations has a group structure, and, at the same time, performs a shift on the parameters of the so-called generalized trigonometric functions introduced by Drábek and Manásevich. Very recently, we have used this class of special functions to express the general solution of the NRT nonlinear Schr\"odinger equation in the free particle case under the traveling wave ansatz.

In later papers, in addition to the application of these new transformationms to higher order autonomous ordinary differential equations, we will further consider their role in the framework of the NRT nonlinear Schr\"odinger equation, i.e., beyond the free particle and constant potential solutions analyzed here, as well as their use in the study of other families of equations of interest. A further relevant question is with respect to finding expressions for the auxiliary field, at least as asymptotic approximations.

\section*{Acknowledgements}

{We gratefully acknowledge the following financial support:
Project PID2020-115273GB-I00 and Grant RED2022-134301-T funded by 
MCIN/ AEI/ 10.13039/501100011033. PRG and AP also gratefully acknowledge 
financial support from the Universidad Rey Juan Carlos as members of the 
grupo de investigación de alto rendimiento DELFO.
\newpage

\appendix

\newpage

\bibliographystyle{unsrt}

\end{document}